\documentclass[a4paper]{article}
\pdfoutput=1
\usepackage{CJK}
\usepackage[paperwidth=185mm,paperheight=230mm,textheight=190mm,textwidth=145mm,left=20mm,right=20mm,
top=25mm, bottom=20mm]{geometry}
\usepackage[CJKbookmarks, colorlinks,
bookmarksnumbered=true,pdfstartview=FitH,linkcolor=blue,citecolor=green]{hyperref}
\usepackage{amsmath,amssymb}
\usepackage{amsthm}
\usepackage{calc}
\usepackage{graphicx}
\usepackage{supertabular}
\usepackage{longtable}
\usepackage{color}
\usepackage{enumerate}
\usepackage{colortbl,booktabs}
\usepackage{multirow}
\usepackage{natbib}
\usepackage{authblk}
\usepackage{threeparttable}

\newtheorem{theorem}{Theorem}
\newtheorem{lemma}{Lemma}

\newtheorem{remark}{Remark}
\newtheorem{assumption}{Assumption}

\def\wh{\widehat}
\def\wt{\widetilde}

\begin{document}

\title{ Robust Functional Principal Component Analysis for Non-Gaussian Longitudinal Data }
\author[a]{{\fontsize{12pt}{18pt}\selectfont Rou Zhong}}
\author[a]{{\fontsize{12pt}{0.5em}\selectfont Shishi Liu}}
\author[a]{{\fontsize{12pt}{0.5em}\selectfont Jingxiao Zhang} \thanks{zhjxiaoruc@163.com}}
\author[b]{{\fontsize{12pt}{0.5em}\selectfont Haocheng Li}}
\affil[a]{{\emph\fontsize{12pt}{0.5em}\selectfont Center for Applied Statistics, School of Statistics, Renmin University of China}}
\affil[b]{{\emph\fontsize{12pt}{0.5em}\selectfont Department of Mathematics and Statistics, University of Calgary}}
\date{}
\maketitle

\begin{abstract}

Functional principal component analysis is essential in functional data analysis, but the inferences will become unconvincing when some non-Gaussian characteristics occur, such as heavy tail and skewness. The focus of this paper is to develop a robust functional principal component analysis methodology in dealing with non-Gaussian longitudinal data, for which sparsity and irregularity along with non-negligible measurement errors must be considered. We introduce a Kendall's $\tau$ function whose particular properties make it a nice proxy for the covariance function in the eigenequation when handling non-Gaussian cases. Moreover, the estimation procedure is presented and the asymptotic theory is also established. We further demonstrate the superiority and robustness of our method through simulation studies and apply the method to the longitudinal CD4 cell count data in an AIDS study.

\end{abstract}

\section{ Introduction }

Functional data analysis has become increasingly important due to its ubiquity in various areas. Many works have been devoted to the development of statistical methods for functional data. We recommend \citet{wang2016functional} for a simple overview of some related studies. Moreover, the monograph by \citet{ramsay2005functional} gave a comprehensive illustration for the major tools in functional data analysis.

With the rapid advancement of functional methodology, longitudinal study under functional framework gradually attracts attention. Longitudinal data are typically measured sparsely and irregularly over time with non-ignorable measurement errors. The dominant methods in longitudinal study are mainly through parametric models, such as random-effects models or generalized estimating equations, for which model assumptions are indispensable. Nevertheless, by regarding longitudinal data as some discrete observations that irregularly sampled from an underlying random process with random fluctuation, nonparametric functional approaches can be exerted. Therefore, less assumptions are needed and greater flexibility can be achieved. For this reason, a great deal of papers explored longitudinal and functional data simultaneously along this line of thought. In term of functional principal component analysis, \citet{yao2005functional} first proposed a technique for longitudinal and functional data through conditional expectation. Subsequently, \citet{hall2008modelling} further considered the latent Gaussian processes, while \citet{jiang2010covariate} accommodated covariate information in the process of functional principal component analysis. \citet{hall2006properties} and \citet{li2010uniform} conducted some theoretical exploration. More complicated cases were taken into consideration in \citet{chiou2014multivariate}, \citet{chiou2016a}, \citet{happ2018multivariate}, in which the subjects are multivariate longitudinal and functional data. As for the functional regression models, to list a few, \citet{jiang2011functional} developed the functional single index models for longitudinal data. \citet{jiang2014inverse} extended the functional sliced inverse regression to accommodate longitudinal covariates. Both \citet{cao2015regression} and \citet{li2020regression} considered the asynchronous data setting issue, that the observation times for the response and covariates are mismatched. Besides functional principal component analysis and regression, some other affairs have also been discussed. \citet{ji2017optimal} constructed optimal designs for longitudinal and functional data by taking into account both trajectory recovery and response prediction. \citet{zhou2018efficient} proposed new estimators for mean and covariance functions using full quasi-likelihood with an adjustment of local kernel smoothing method.

Among the aforementioned methodology in functional data analysis, functional principal component analysis is a crucial technique due to the intrinsic infinite dimensional feature of functional data. Unfortunately, as pointed out by \citet{kraus2012dispersion} and \citet{zhong2020cluster}, the inferences through functional methodology, especially functional principal component analysis, may be implausible when the data deviate from Gaussian assumption. Therefore, some robust estimators for principal components have been explored in the existing literatures. \citet{locantore1999robust} introduced the spherical principal component analysis, the intuition of which is projecting the data onto the unit sphere. Following this idea, \citet{gervini2008robust} further developed the spherical principal components based on a newly defined functional median. In addition, theoretical properties and robustness of their estimators were also involved in. More asymptotic results were given in \citet{boente2019the}. \citet{bali2011robust} proposed several robust function principal component estimators through a projection-pursuit approach. \citet{boente2015s} suggested the S-estimators for functional principal components using robust scale estimates. Although the Gaussian assumption is not essential any more for the above results, some distributional constraints are still imposed, such as symmetric or elliptical. On the other hand, they merely considered densely observed functional data and the measurement error was scarcely involved. To the best of our knowledge, only \citet{hall2008modelling} discussed non-Gaussian principal component estimates for longitudinal data. They proposed a latent Gaussian model, for which a known link function is required. Consequently, some implicit constraints on distribution are also imposed.

The goal of this paper is to provide robust estimators for functional principal components, which can be distribution-free and accommodate to longitudinal data. For this purpose, we propose a Kendall functional principal component analysis approach. To be specific, a novel Kendall's $\tau$ function is defined, which is inspired by the Kendall's $\tau$ correlation coefficient in \citet{kendall1938a} and  the spatial sign covariance function in \citet{gervini2008robust} and \citet{boente2019the}. The Kendall's $\tau$ function is further shown to have the same eigenspace as the population covariance function, with no need of assumption of symmetric distribution which is crucial for \citet{gervini2008robust} and \citet{boente2019the}. For the computation, local linear estimate for Kendall's $\tau$ function is introduced, then the estimations of the functional principal components are followed afterward. As the rank information is utilized, robustness can be achieved. It coincides with our simulation results which indicate that our method performs well even under heavy-tailed or skewed distribution for both dense and sparse designs. Moreover, we also explore the asymptotic consistency properties of our estimators in detail.

The contributions of this paper are summarized as follow. First, we proposed the Kendall's $\tau$ function whose principal components are identical to those of covariance function. Compared with covariance function, Kendall's $\tau$ function is less likely to be affected by the violation of Gaussian assumption, so that even heavy-tailed and skewed cases can be handled. Second, we introduced an estimating procedure for Kendall's $\tau$ function that suits for longitudinal data. Thereby the principal component estimators are subsequently obtained from the estimated Kendall's $\tau$ function. Third, we confirmed the asymptotic consistency for the estimated Kendall's $\tau$ function and principal components through careful theoretical analysis.

\section{ Methodology }\label{KFPCA}

\subsection{Kendall Functional Principal Component Analysis}\label{seckendfun}

Let $X(t)$ be a random process in $L^2(\mathcal{T})$ with mean function $\mu(t) = E\{X(t)\}$ and covariance function $\Sigma(s, t) = \mbox{cov}\{X(s), X(t)\}$, where $\mathcal{T}$ is a bounded and closed interval. The Mercer's Theorem provides spectral decomposition of functional version for $\Sigma(s, t)$, that is $\Sigma(s, t) = \sum_{k = 1}^{\infty} \lambda_k \phi_k(s) \phi_k(t)$, where $\lambda_1 > \lambda_2 > \cdots$ are the eigenvalues and $\phi_k$'s are the corresponding eigenfunctions which form an orthogonal basis on $L^2(\mathcal{T})$. Further, $X(t)$ admits the expansion $X(t) = \mu(t) + \sum_{k = 1}^{\infty} \xi_k \phi_k(t)$, where the uncorrelated random coefficients $\xi_k, k = 1, 2, \ldots$ are functional principal component scores with mean zero and variance $\lambda_k$. The expansion refers to as the well-known Karhunen-Lo\`{e}ve Theorem and indicates that the data-driven basis $\{\phi_k\}$ is optimal in the sense of mean squared error. In addition, the eigenfunctions, especially the first few, possess the capacity of characterizing various variation modes of the random function $X$ \citep{hall2006properties}. Due to the merits above, estimation of the eigenfunctions becomes the primary task in functional principal component analysis and is of great interest.

The standard functional principal component analysis approach is based on solving the following eigenequations
\begin{align}
\int_{\mathcal{T}} \Sigma(s, t) \phi_k(t) dt = \lambda_k \phi_k(s), k = 1, 2, \ldots. \nonumber
\end{align}
Thus, an efficient estimation of the covariance function $\Sigma(s, t)$ turns to be the key issue in some articles related to functional principal component analysis \citep{lin2020mean, wang2020low}. However, covariance structure is unlikely to be accurately inferred in non-Gaussian circumstances \citep{kraus2012dispersion}. In order to improve estimating effectiveness for non-Gaussian cases, we conduct our analysis in a different way and propose a more flexible approach.

To start with, \citet{gervini2008robust} introduced a weighted covariance function (called spatial sign covariance function) which can be regarded as the covariance function of the centered curve after projecting onto the unit sphere. They demonstrated that the weighted covariance function shares the same eigenfunctions with the population covariance function. Nevertheless, this assertion was established under the assumption of symmetric distribution. Motivated by their work, we define the Kendall's $\tau$ function as
\begin{align}
K(s, t) = E\Big [ \frac{\{X(s) - \wt{X}(s)\}\{X(t) - \wt{X}(t)\}}{\|X - \wt{X}\|^2} \Big ], \nonumber
\end{align}
where $\wt{X}$ is an independent copy of $X$ and $\|X - \wt{X}\|^2 = \int \{X(u) - \wt{X}(u)\}^2 du$. Theorem \ref{theorykendall} in Section \ref{proKend} shows that the Kendall's $\tau$ function $K(s, t)$ also has the same eigenspace as $\Sigma(s, t)$, but without any distributional constraints. Properties of $K(s, t)$ are carefully discussed in Section \ref{proKend}.

The Kendall's $\tau$ function involves the comparison of two subjects, $X$ and $\wt{X}$, while \citet{gervini2008robust} considered contrasting $X$ with its functional median. In fact, this is the pivotal reason for the superiority of Kendall's $\tau$ function, of which the construction borrows the idea of Kendall's $\tau$ correlation coefficient. More intuition can be obtained from our proof of Theorem \ref{theorykendall} in the Supplementary Material. Moreover, while \citet{gervini2008robust} assumed a finite $p$-component model, our results are even suitable for the infinite series expansions.

Consequently, Kendall functional principal component analysis considers the following eigenequations,
\begin{align}
\int_{\mathcal{T}} K(s, t) \phi_k(t) dt = \lambda_k^{*} \phi_k(s), k = 1, 2, \ldots, \nonumber
\end{align}
where $\phi_k$'s are the targeted eigenfunctions and $\lambda_k^{*}$ is the $k$th eigenvalue of $K(s, t)$ which may not be equivalent to $\lambda_k$. Therefore, the estimation of $K(s, t)$ plays a crucial role in implementing our methods and the estimation procedure is thoroughly investigated in Section \ref{estimation}.

\subsection{Estimation}\label{estimation}

To fit in with the characteristics of the longitudinal data, an extended model with additional measurement error is introduced. Denote $X_i$'s as the independent realizations of $X$. Let $t_{ij} \in \mathcal{T}, i = 1, \ldots, N, j = 1, \ldots , m_i$ be the $j$th observation time for the $i$th subject, where $N$ is the sample size and $m_i$ is the observation size for the $i$th subject. We assume that
\begin{align}\label{model}
Y_{ij} = X_i(t_{ij}) + \epsilon_{ij},
\end{align}
where $Y_{ij}$ is the observation of the $i$th subject at $t_{ij}$ with independent measurement error $\epsilon_{ij}$. The errors are independent of $X_i$'s and identically distributed with $E (\epsilon_{ij}) = 0$ and $\mbox{var}(\epsilon_{ij}) = \sigma^2$. Model (\ref{model}) is commonly used in the study for both longitudinal and functional data, and based on the Karhunen-Lo\`{e}ve expansion, we have
\begin{align}\label{modelkl}
Y_{ij} = \mu(t_{ij}) + \sum_{k = 1}^{\infty} \xi_{ik} \phi_k(t_{ij}) + \epsilon_{ij},
\end{align}
where $\xi_{ik}$ is the $k$th functional principal component score for the $i$th subject.

Starting from (\ref{model}), we apply local linear smoother approach to handle the estimation of $K(s, t)$. It is known that local linear smoother tends to pool all data across subjects together to achieve satisfactory estimates, thus is widely used in estimating the mean function and covariance function when Model (\ref{model}) is under consideration \citep{yao2005functional, li2010uniform}. However, the application of local linear smoother in estimating $K(s, t)$ is nontrivial as it involves $X$ and $\wt{X}$ simultaneously. In truth, Model (\ref{model}) implies that each sample measured discretely and the observation time points are distinct from each other. The irregularity and the possible sparseness challenge the estimating procedure.

Let $k(\cdot)$ be a kernel function on $[-1, 1]$ and $k_{h^{'}}(t) = k(t/h^{'})/h^{'}$, where $h^{'}$ is the bandwidth. Define the ``raw" Kendall's $\tau$ covariances as
\begin{align}\label{rawKend}
K_i(t_{ik}, t_{il}) = \frac{1}{N-1} \sum_{j \neq i} \Big [ \frac{\big \{ Y_{ik} - \wh{X}_{j}(t_{ik}) \big \} \big \{ Y_{il} - \wh{X}_{j}(t_{il}) \big \}}{\sum_{q = 1}^{m_i} \{ Y_{iq} - \wh{X}_j(t_{iq})\}^2 / m_i} \Big ],
\end{align}
where
\begin{align}
\wh{X}_{j}(t_{ik}) = \frac{\sum_q k_{h^{'}}(t_{ik} - t_{jq}) Y_{jq}}{\sum_q k_{h^{'}}(t_{ik} - t_{jq})}. \nonumber
\end{align}
Then taking $K_i(t_{ik}, t_{il})$ as the input data for the local linear smoothing. More specifically, we endeavor to minimize
\begin{align}\label{loclin}
\sum_{i = 1}^N \sum_{1 \leq k \neq l \leq m_i} \kappa_h \big ( t_{ik} - s \big ) \kappa_h \big ( t_{il} - t \big ) \{ K_i(t_{ik}, t_{il}) - \beta_0 - \beta_{11}(s - t_{ik}) - \beta_{12}(t - t_{il})\}^2
\end{align}
with respect to $\beta_0$, $\beta_{11}$ and $\beta_{12}$, where $\kappa_h \big ( t \big ) = \kappa \big ( t/h \big )/h $, $\kappa(\cdot)$ is a kernel function and $h$ is the bandwidth. We further define $K_0(s, t) = |\mathcal{T}| K(s, t)$, where $|\mathcal{T}|$ represents the length of the interval $\mathcal{T}$, then $\wh{\beta}_0$ is the estimate of $K_0(s, t)$, so we also denote it as $\wh{K}_0(s, t)$. We only consider the pairs $(t_{ik}, t_{il})$ for $k \neq l$ in (\ref{loclin}) to avoid the influence of the measurement error on the diagonal.

\begin{remark}\label{remarkrawkend}
For the construction of the ``raw" Kendall's $\tau$ covariances $K_i(t_{ik}, t_{il})$ in (\ref{rawKend}), the estimation of $X_j(t_{ik}), j \neq i$ is drawn to settle the irregularity issue. To estimate $X_j(t_{ik})$, we introduce a weighted average of $\{Y_{j1}, \ldots, Y_{jm_j}\}$ based on the local kernel weighting techniques, the underlying idea of which is offering more weight to the observations close in time while assigning less weight to the distant observations. The kernel function $k(\cdot)$ with support $[-1, 1]$ is employed to the estimation. In practical, for the extremely sparse data, there may be no observation falls within the window width, then the estimation of $X_j(t_{ik})$ is infeasible for some $j$. In this case, we directly neglect the comparison of the certain $j$th sample with the $i$th sample in the computation of $K_i(t_{ik}, t_{il})$. Therefore, the bandwidth $h^{'}$ is observation size related, which can neither be very large to bring inaccuracy, nor very small to eliminate too much comparisons.
\end{remark}

\begin{remark}\label{remarkloclin}
By minimizing (\ref{loclin}), we obtain the estimate of $K_0(s, t)$ rather than the Kendall's $\tau$ function $K(s, t)$. However, $K_0(s, t)$ and $K(s, t)$ possess the same eigenfunctions since they differ only by a constant multiple. Hence, it is plausible to estimate the eigenfunctions via $\wh{K}_0(s, t)$ subsequently, more computational details for eigenfunctions are stated later. In addition, we adopt the generalized cross-validation method for the selection of the bandwidth $h$ in (\ref{loclin}).
\end{remark}

From the discussion in Remark \ref{remarkloclin}, we know that $\phi_k$'s are exactly the eigenfunctions of $K_0(s, t)$. We further denote the $k$th eigenvalue of $K_0(s, t)$ as $\rho_k, k = 1, 2, \ldots$. Then given $\wh{K}_0(s, t)$, the estimated eigenfunctions $\wh{\phi}_k$'s are obtained from the eigenequations
\begin{align}\label{esteigeneq}
\int_{\mathcal{T}} \wh{K}_0(s, t) \wh{\phi}_k(t) dt = \wh{\rho}_k \wh{\phi}_k(s), k = 1, 2, \ldots,
\end{align}
where $\wh{\rho}_k$ is the estimate of $\rho_k$, $\wh{\phi}_k$ satisfies $\int_{\mathcal{T}} \wh{\phi}_k (t)^2 dt = 1$ and $\int_{\mathcal{T}} \wh{\phi}_k(t) \wh{\phi}_j(t) dt = 0$ for any $j \neq k$. Refer to \citet{ramsay2005functional}, the approximate solutions of eigenequations (\ref{esteigeneq}) can be acquired by discretizing $\wh{K}_0(s, t)$.

\section{Theory}\label{theory}

\subsection{Properties of Kendall's $\tau$ function}\label{proKend}

In this section, we discuss properties of the Kendall's $\tau$ function $K(s, t)$ in population view. Recall the Karhunen-Lo\`{e}ve expansion
\begin{align}\label{KL}
X(t) = \mu(t) + \sum_{k = 1}^{\infty} \xi_k \phi_k(t),
\end{align}
which is mentioned in Section \ref{seckendfun}. We do not impose any distributional assumptions for $\xi_k, k = 1, 2, \ldots$ here. Therefore, the following attractive properties hold even for those distributions with peculiar characteristics, such as heavy-tailed and skewed.

\begin{theorem}\label{theorykendall}
Denote the $k$th eigenfunction of Kendall's $\tau$ function $K(s, t)$ as $\phi_k^{*}$. Under (\ref{KL}), we have $\phi_k^{*} = \phi_k, k = 1, 2, \ldots$. Additionally, the eigenvalues $\{\lambda_k^{*}\}_{k = 1}^{\infty}$ of $K(s, t)$ have the same order as $\{\lambda_k\}_{k = 1}^{\infty}$, which means $\lambda_1^{*} > \lambda_2^{*} > \cdots$, if $\lambda_1 > \lambda_2 > \cdots$.
\end{theorem}

Theorem \ref{theorykendall} illustrates the potential of $K(s, t)$ in functional principal component analysis, for eigenfunctions of $K(s, t)$ are the same as those of covariance function. Further, $K(s, t)$ is less affected by distribution of the data than covariance function and spatial sign covariance function in \citet{gervini2008robust}. Besides, more insights of the exact relationship between $\lambda_{k}^{*}$ and $\lambda_k$ can be gained from the proof of Theorem \ref{theorykendall} in the Supplementary Material.

\subsection{Consistency}

We present the asymptotic properties of $\wh{K}_0(s, t)$ and the estimated eigenfunctions $\wh{\phi}_k$'s in this section. Proofs are provided in the Supplementary Material. We first list the assumptions below.
\begin{assumption}\label{kernel}
The kernel functions $k(\cdot)$ and $\kappa(\cdot)$ are symmetric, bounded and compactly supported on $[-1, 1]$. Moreover, $\int v^2\kappa(v) dv= \nu_2$, where $\nu_2$ is a nonzero scalar.
\end{assumption}
\begin{assumption}\label{TYm}
The observation number $m_i$'s are the independent realizations of the random variable $m$, and are independent of $\big \{(t_{ij}, Y_{ij}): j = 1, \ldots, m_i \big \}$.
Assume all $m_i$'s have the same order, that is $m_i \sim O_p(M), i = 1, \ldots, N$ and $M = O(N^{\alpha}), \alpha > 0$.
\end{assumption}
\begin{assumption}\label{K0assum}
$K_0(s, t)$ is twice continuously differentiable on $\mathcal{T}^2$ and its second-order partial derivatives are all bounded.
\end{assumption}
\begin{assumption}\label{density}
Let $g(t)$ be the density function of the observation times $t_{ij}$. Assume that $g(t)$ is bounded away from zero with continuous and bounded second derivative.
\end{assumption}
\begin{assumption}\label{Xder}
$X(t)$ is smooth over $\mathcal{T}$ with bounded first derivative.
\end{assumption}
\begin{assumption}\label{deltat}
For the $i$th subject, let $\{t_1, \ldots, t_{m_i}\}$ be the equally spaced time points on $\mathcal{T}$. Assume that $t_{iq} - t_{q} = O_p(\Delta t), q = 1, \ldots, m_i, i = 1, \ldots, N$ and $\lim_{N \rightarrow \infty} \Delta t = 0$.
\end{assumption}
\begin{assumption}\label{bandwidth}
The bandwidths satisfy $\lim_{N \rightarrow \infty} h^{'} = 0$, $\lim_{N \rightarrow \infty} Mh^{'} < \infty$, $\lim_{N \rightarrow \infty} h = 0$, $\lim_{N \rightarrow \infty} NM^2h^2/\log N = \infty$ and $\lim_{N \rightarrow \infty} NM^2h^6/\log N < \infty$.
\end{assumption}
Here Assumption \ref{kernel} is a general assumption for kernel functions. Assumption \ref{TYm} requires the observation numbers to be independent of the observation times and the observations \citep{yao2005functional}. Further, it also requests the magnitude of the observation size to be some order of the sample size so that the degree of sparseness can be described \citep{li2010uniform, cai2011optimal, zhang2016from}. Assumptions \ref{K0assum}-\ref{Xder} are conventional assumptions in functional data analysis. Similar to \citet{hall2008modelling}, we also demand the observation times are approximately uniformly laid over the interval $\mathcal{T}$ in Assumption \ref{deltat}. The bandwidth conditions are stated in Assumption \ref{bandwidth}.

\begin{theorem}\label{consistency}
Under Assumptions \ref{kernel}-\ref{bandwidth}, we have
\begin{align}\label{consisform}
\sup_{s,t\in \mathcal{T}} \Big |\wh{K}_0(s, t) - K_0(s, t) \Big | = O_p \Big ( \tau_1 + \tau_2 \Big ),
\end{align}
where $\tau_1 = \big (NM^2h^2/\log N \big )^{-1/2} + h^2$ and $\tau_2 = M^{-1} + h^{'}$.
\end{theorem}

\begin{theorem}\label{eigenfunction}
Under Assumptions \ref{kernel}-\ref{bandwidth}, we have
\begin{align}\label{eigsup}
\sup_{t \in \mathcal{T}} \big | \wh{\phi}_k(t) - \phi_k(t) \big | = O_p \Big ( \tau_1 + \tau_2 \Big ),
\end{align}
where $\tau_1 = \big (NM^2h^2/\log N \big )^{-1/2} + h^2$ and $\tau_2 = M^{-1} + h^{'}$.
\end{theorem}

\begin{remark}
From (\ref{consisform}) and (\ref{eigsup}), $\wh{K}_0(s, t)$ and $\wh{\phi}_k$'s share the same convergence rate which can be separated into two parts, $\tau_1$ and $\tau_2$. The first term $\tau_1$ is a usual convergence rate for local linear smoother \citep{li2010uniform, li2011efficient}, while the second term $\tau_2$ originates from the estimation mentioned in Remark \ref{remarkrawkend}, and the estimation error is inevitable for the irregularity and sparsity of the observation times. It has been recognized that various magnitudes of the observation size $M$ may lead to distinct asymptotic results \citep{cai2011optimal, zhang2016from}. Such feature occurs in our results with no exception. For dense data where $\alpha$ is relatively large, the convergence rates are dominated by $\tau_1$. In contrast, for a extremely small $\alpha$, which results in sparse data, the convergence rate will be somewhat slow. Further, we conduct the numerical experiments for both dense and sparse cases in Section \ref{simulation} and we delightedly find that our method outperforms the other method in both circumstances.
\end{remark}

\section{ Simulation }\label{simulation}

To illustrate the performances of our method, we execute abundant simulation studies considering various eigenfunction settings and distinct distributions for both dense and sparse design. To be specific, the data are generated from (\ref{modelkl}) with $\mu(t) = t + \sin(t)$, $\epsilon_{ij} \sim N(0, 0.1)$. We further set the eigenvalues as $\lambda_1 = 9$, $\lambda_2 = 1.5$ and $\lambda_k = 0, k \geq 3$. For the eigenfunctions, we consider the following four cases:
\begin{itemize}
\item Case 1: $\phi_1(t) = \mbox{cos}(\pi t/10)/\sqrt{5}$ and $\phi_2(t) = \mbox{sin}(\pi t/10)/\sqrt{5}, t \in [0,10]$.
\item Case 2: $\phi_1(t) = \mbox{cos}(\pi t/10)/\sqrt{5}$ and $\phi_2(t) = \mbox{cos}(\pi t/5)/\sqrt{5}, t \in [0,10]$.
\item Case 3: $\phi_1(t) = \mbox{cos}(\pi t/5)/\sqrt{5}$ and $\phi_2(t) = \mbox{sin}(\pi t/5)/\sqrt{5}, t \in [0,10]$.
\item Case 4: $\phi_1(t) = \mbox{cos}(\pi t/5)/\sqrt{5}$ and $\phi_2(t) = \mbox{sin}(2\pi t/5)/\sqrt{5}, t \in [0,10]$.
\end{itemize}
It can be realized that from Case 1 to Case 4, more volatility is provided to the eigenfunctions because the periods become shorter, leading to more tough estimation.

To demonstrate the robustness of our method to data deviating from Gaussian distribution, four distinct distributions are taken into account, by setting the component scores $\xi_{ik}$'s to follow Gaussian, mix-Gaussian \citep{yao2005functional}, EC2 \citep{han2018eca}, skew-t \citep{azzalini2014the} distributions respectively with zero mean and corresponding variance. To give more detail, EC2 distribution is an elliptical distribution with heavy-tailed feature while skew-t distribution is an asymmetric distribution with skewed feature.

The irregularity and sparsity are set up similar to \citet{yao2005functional}. Let $\{c_0, \ldots, c_{50}\}$ be the equally spaced grid over $[0, 10]$ with $c_0 = 0$ and $c_{50} = 10$. Further, let $s_i = c_i + e_i, i = 0, \ldots, 50$ with $e_i \sim N(0, 0.1)$ and if $s_i < 0$ or $s_i > 10$, we set it as 0 and 10 respectively. So far the irregularity is imposed. For the observation number, we randomly choose it from $\{8, \ldots, 12\}$ for the relatively dense design and $\{2, \ldots, 5\}$ for the sparse design. Then the observation locations are randomly chosen from $\{s_1, \ldots, s_{49}\}$ in term of the observation numbers. Here we disregard the case in which some subjects with only one observation, for those subjects bring no contribution to the estimation. Moreover, we conduct $100$ simulation runs with the sample size $N = 100$ in each run.

To compare our method with the popular method (PACE) in \citet{yao2005functional}, we consider the following criteria:
\begin{align}
\mbox{IMSE}_i &= \|\phi_i - \wh{\phi}_i\|^2 = \int_{\mathcal{T}} \{ \phi_i(t) - \wh{\phi}_i(t)\}^2 dt, i = 1, 2, \nonumber \\
\mbox{Angle}_i &= \mbox{arccos}\Big \{ \Big |\int_{\mathcal{T}} \phi_i(t) \wh{\phi}_i(t) dt \Big | \Big \}, i = 1,2, \nonumber
\end{align}
where $\mbox{IMSE}_i$ is the integrated mean squared error of $\wh{\phi}_i$ and $\mbox{Angle}_i$, which ranges from 0 to 90, is the angle between the direction of the true $\phi_i$ and the estimated $\wh{\phi}_i$. They all measure the discrepancies of the eigenfunction estimates.

The simulation results under dense design are summarized in Table \ref{simdense}. It is evident that across all the four eigenfunctions settings, the proposed method outperforms method of \citet{yao2005functional} with greatly smaller $\mbox{IMSE}_i$ and $\mbox{Angle}_i$, even under Gaussian distribution. That means more precise estimation can be achieved through Kendall functional principal component analysis. Further, we observe that the $\mbox{IMSE}_i$ and $\mbox{Angle}_i$ of Kendall functional principal component analysis under Gaussian distribution exhibit a slight increase when turning to EC2 distribution and skew-t distribution, while method by \citet{yao2005functional} endures an obvious increasing trend, which demonstrates the robustness of our approach. The above findings provide numerical evidence for illustrating the superiority of our method in functional principal component analysis and its robustness to non-Gaussian data. Moreover, a peculiar decrease of the estimation error is emerged when switching from Gaussian distribution to mix-Gaussian distribution, which can also be found in the simulation results of \citet{yao2005functional}. It is interesting that mix-Gaussian data can reach more accurate results than Gaussian data.

\begin{table}
\caption{Results of simulation studies under dense design settings}
\label{simdense}
\begin{center}
\setlength{\tabcolsep}{1mm}{
\begin{threeparttable}
\begin{tabular}{cccccccccc}
\hline
 & &\multicolumn{2}{c}{Gaussian}&\multicolumn{2}{c}{mix-Gaussian}&\multicolumn{2}{c}{EC2}&\multicolumn{2}{c}{Skew-t}\\
 & &PACE&KFPCA&PACE&KFPCA&PACE&KFPCA&PACE&KFPCA\\
\hline
\multirow{4}{*}{Case 1}&$\mbox{IMSE}_1$&0.013&0.006&0.011&0.006&0.018&0.007&0.030&0.008\\
 &$\mbox{IMSE}_2$&0.062&0.018&0.045&0.017&0.0780&0.018&0.119&0.019\\
 &$\mbox{Angle}_1$&6.198&4.091&5.653&4.111&7.034&4.396&8.425&4.493\\
 &$\mbox{Angle}_2$&11.716&7.294&11.188&7.257&13.671&7.241&16.065&7.662\\
\hline
\multirow{4}{*}{Case 2}&$\mbox{IMSE}_1$&0.015&0.008&0.013&0.007&0.025&0.007&0.024&0.009\\
 &$\mbox{IMSE}_2$&0.082&0.027&0.076&0.023&0.132&0.030&0.132&0.027\\
 &$\mbox{Angle}_1$&6.460&4.844&6.182&4.482&8.112&4.569&8.161&4.978\\
 &$\mbox{Angle}_2$&15.455&8.963&14.460&8.341&18.590&9.377&18.521&9.144\\
\hline
\multirow{4}{*}{Case 3}&$\mbox{IMSE}_1$&0.017&0.009&0.014&0.007&0.026&0.011&0.034&0.008\\
 &$\mbox{IMSE}_2$&0.109&0.025&0.079&0.025&0.172&0.031&0.206&0.026\\
 &$\mbox{Angle}_1$&7.236&5.081&6.551&4.608&8.779&5.414&9.947&4.848\\
 &$\mbox{Angle}_2$&16.888&8.801&15.320&8.676&20.549&9.576&22.498&8.825\\
\hline
\multirow{4}{*}{Case 4}&$\mbox{IMSE}_1$&0.019&0.008&0.015&0.006&0.028&0.008&0.031&0.007\\
 &$\mbox{IMSE}_2$&0.184&0.046&0.198&0.041&0.430&0.052&0.426&0.053\\
 &$\mbox{Angle}_1$&7.488&4.874&6.695&4.297&8.778&5.019&9.221&4.633\\
 &$\mbox{Angle}_2$&22.757&11.789&23.207&10.986&33.464&12.508&34.392&12.606\\
\hline
\end{tabular}
\begin{tablenotes}
\footnotesize
\item[] PACE, principal component analysis through conditional expectation \citep{yao2005functional}.
\item[] KFPCA, the proposed Kendall functional principal component analysis.
\end{tablenotes}
\end{threeparttable}}
\end{center}
\end{table}

Table \ref{simsparse} shows the simulation results under sparse design for various distributions across the four eigenfunctions settings. For the sparse design, the $\mbox{IMSE}_i$ and $\mbox{Angle}_i$ of both methods show a significant growth compared with those in the dense design, as a result of less information contained in the former. As for the comparison of these two methods, we can get similar conclusions from Table \ref{simsparse} as the dense design. To be specific, our method leads to slightly better estimates for the first eigenfunction, while substantial improvements are gained for the estimation of the second eigenfunction. In contrast with the first eigenfunction, some deterioration in estimating the second eigenfunction is shown for both methods. We conjecture that the reason lies in the extreme sparseness, which makes the estimation for the higher-order eigenfunctions much tough. Furthermore, the robustness of Kendall functional principal component analysis can be revealed in the same way as the dense design.

\begin{table}
\caption{Results of simulation studies under sparse design settings}
\label{simsparse}
\begin{center}
\setlength{\tabcolsep}{1mm}{
\begin{threeparttable}
\begin{tabular}{cccccccccc}
\hline
 & &\multicolumn{2}{c}{Gaussian}&\multicolumn{2}{c}{mix-Gaussian}&\multicolumn{2}{c}{EC2}&\multicolumn{2}{c}{Skew-t}\\
 & &PACE&KFPCA&PACE&KFPCA&PACE&KFPCA&PACE&KFPCA\\
\hline
\multirow{4}{*}{Case 1}&$\mbox{IMSE}_1$&0.030&0.020&0.027&0.019&0.057&0.024&0.051&0.026\\
 &$\mbox{IMSE}_2$&0.247&0.010&0.251&0.102&0.296&0.113&0.300&0.129\\
 &$\mbox{Angle}_1$&9.054&7.373&8.511&7.234&11.852&8.096&11.232&8.442\\
 &$\mbox{Angle}_2$&25.436&16.649&25.097&16.952&27.735&17.470&27.185&19.138\\
\hline
\multirow{4}{*}{Case 2}&$\mbox{IMSE}_1$&0.036&0.023&0.025&0.029&0.047&0.032&0.051&0.036\\
 &$\mbox{IMSE}_2$&0.313&0.212&0.260&0.237&0.468&0.292&0.432&0.274\\
 &$\mbox{Angle}_1$&9.678&7.965&8.107&9.039&11.062&9.438&11.512&9.937\\
 &$\mbox{Angle}_2$&28.997&23.397&27.253&25.445&36.680&28.040&35.713&26.912\\
\hline
\multirow{4}{*}{Case 3}&$\mbox{IMSE}_1$&0.055&0.043&0.047&0.039&0.097&0.058&0.111&0.060\\
 &$\mbox{IMSE}_2$&0.512&0.184&0.429&0.190&0.553&0.220&0.530&0.229\\
 &$\mbox{Angle}_1$&12.619&11.301&11.757&10.717&16.311&12.966&17.034&13.021\\
 &$\mbox{Angle}_2$&38.519&22.623&34.593&23.099&39.885&25.037&39.776&25.821\\
\hline
\multirow{4}{*}{Case 4}&$\mbox{IMSE}_1$&0.062&0.039&0.038&0.032&0.090&0.045&0.098&0.046\\
 &$\mbox{IMSE}_2$&0.892&0.664&0.890&0.678&1.039&0.763&1.004&0.681\\
 &$\mbox{Angle}_1$&13.181&10.614&10.590&9.696&15.658&11.432&16.375&11.385\\
 &$\mbox{Angle}_2$&54.662&46.060&54.715&45.870&59.781&49.546&58.799&46.877\\
\hline
\end{tabular}
\begin{tablenotes}
\footnotesize
\item[] PACE, principal component analysis through conditional expectation \citep{yao2005functional}.
\item[] KFPCA, the proposed Kendall functional principal component analysis.
\end{tablenotes}
\end{threeparttable}}
\end{center}
\end{table}

\section{ Real Data Analysis }\label{real}

In this section, we apply the proposed method on the longitudinal CD4 cell count data from a human immunodeficiency virus (HIV) study for patients with AIDS \citep{wohl2005cytomegalovirus}. The data set can be obtained from \citet{cao2015regression}. As we know, CD4 cells play an important role in our immune system. However, CD4 cells will be attacked and destroyed when infected HIV. Therefore, the number or percentage of CD4 cells reflect the disease progression for the HIV infected person, and then become a usual health status marker. There were $190$ subjects enrolled in this AIDS study from June 1997 to January 2002. Nevertheless, because of the occurrences of missing visits and the randomness of HIV infections, the data evidently show some sparse and irregular features. Figure \ref{designplot}(a) is the design plot for the entire data. Although the data are not required to be dense, the pooled data from all subjects should roughly fill the domain of $\mathcal{T} \times \mathcal{T}$. It is obvious that the assembled pairs $(t_{ij}, t_{ik}), i = 1, \ldots, N, j \neq k$ are extremely sparse in the first $400$ days, making both the proposed method and method by \citet{yao2005functional} unstable. Therefore, the measurements of the first $400$ days are neglected in the following analysis and the design plot after adjustment is given in Figure \ref{designplot}(b). Moreover, the frequencies of different observation numbers are listed in Table \ref{obsfreq}. As we only consider the data after the first $400$ days, one subject is automatically eliminated for he/she did not have any record afterwards.

\begin{figure}[htbp]
\centering
\includegraphics[width=\textwidth]{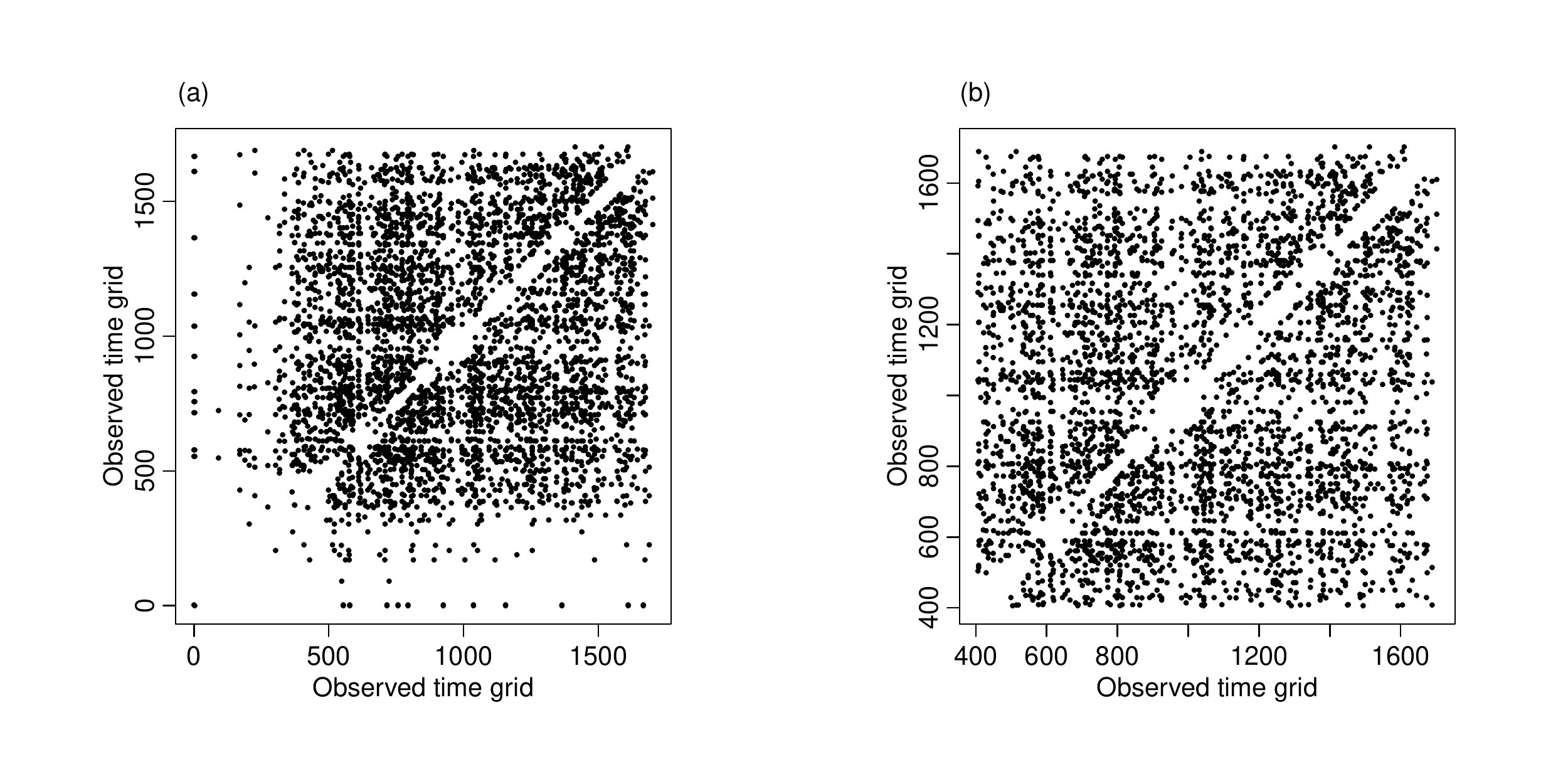}\\
\caption{(a) Design plot for the entire CD4 count data. (b) Design plot for the CD4 count data after the first $400$ days. Design plot is the plot for the assembled pairs $(t_{ij}, t_{ik})$, where $i = 1, \ldots, N, j \neq k$.}
\label{designplot}
\end{figure}

\begin{table}[htbp]
\caption{The frequency of different observation numbers for the CD4 count data.}
  \label{obsfreq}
  \begin{center}
  \begin{tabular}{cc|cc}
  \hline
  Observation Number&Frequency&Observation Number&Frequency \\
  \hline
  1&44&8&8\\
  2&39&9&5\\
  3&29&10&3\\
  4&25&11&2\\
  5&12&12&2\\
  6&9&13&2\\
  7&8&15&1\\
  \hline
  \end{tabular}
  \end{center}
\end{table}

Refer to \citet{cao2015regression}, we first log-transformed the CD4 count data. The kernel density estimate of the log-transformed CD4 count data is shown in Figure \ref{densityplot}. The visible heavy-tailed and skewed features indicate a severely deviation from Gaussian distribution. That means the classical method may not be appropriate for the analysis. Before implementing our method, we estimate the mean function for the log-transformed CD4 count data using local linear smoother. The estimated mean function is displayed in Figure \ref{meaneig}(a), which exhibits a roughly increasing trend of the CD4 cell count. Figures \ref{meaneig}(b)(c) give the estimation of the first two eigenfunctions through Kendall functional principal component analysis. We find that the first eigenfunction approximately reveals the increasing trend of the mean function at the early stage, while the second eigenfunction highlights the contrast between the early stage and the late stage.

\begin{figure}[htbp]
\centering
\includegraphics[width=0.6\textwidth]{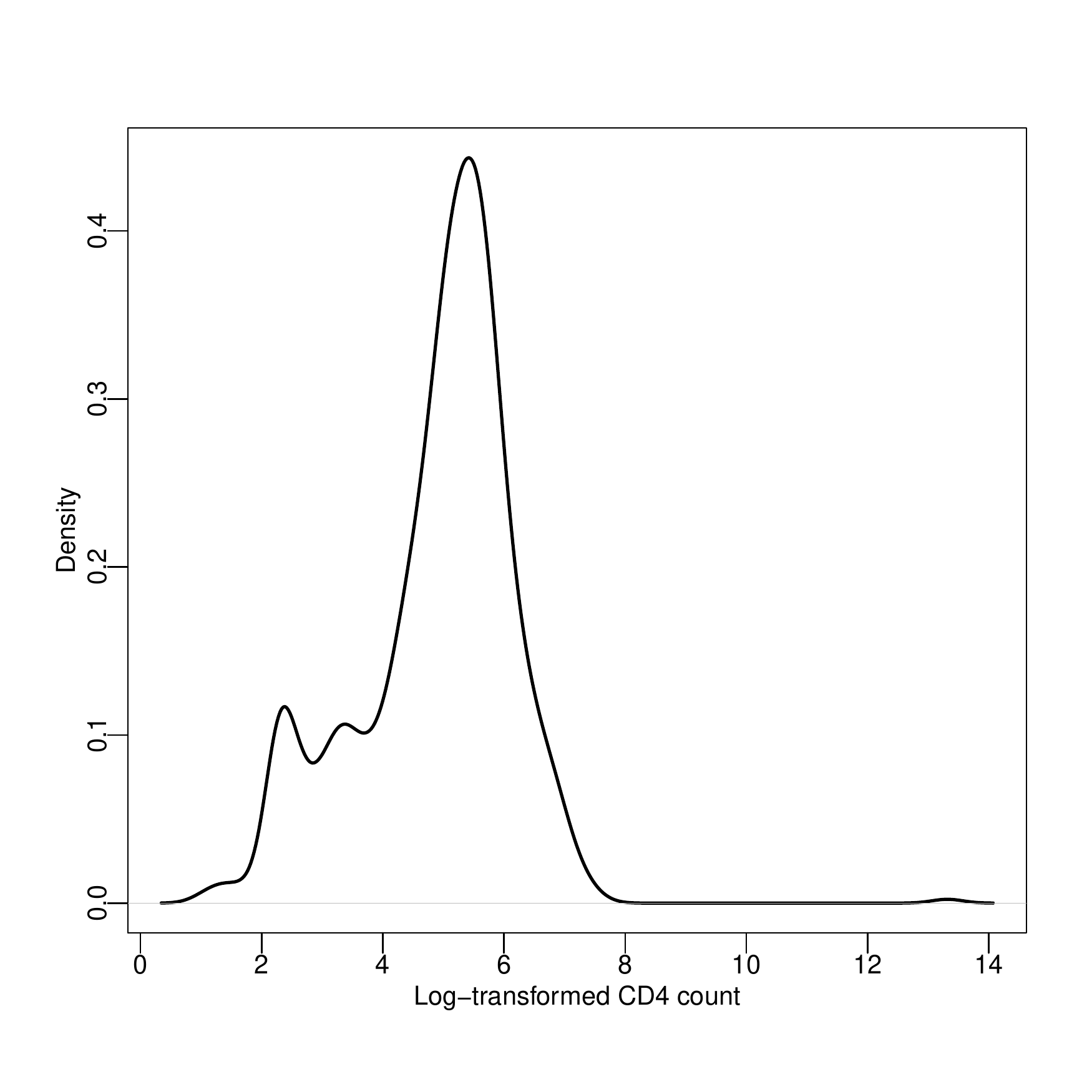}\\
\caption{The kernel density estimate of the log-transformed CD4 count data.}
\label{densityplot}
\end{figure}

\begin{figure}[htbp]
\centering
\includegraphics[width=\textwidth]{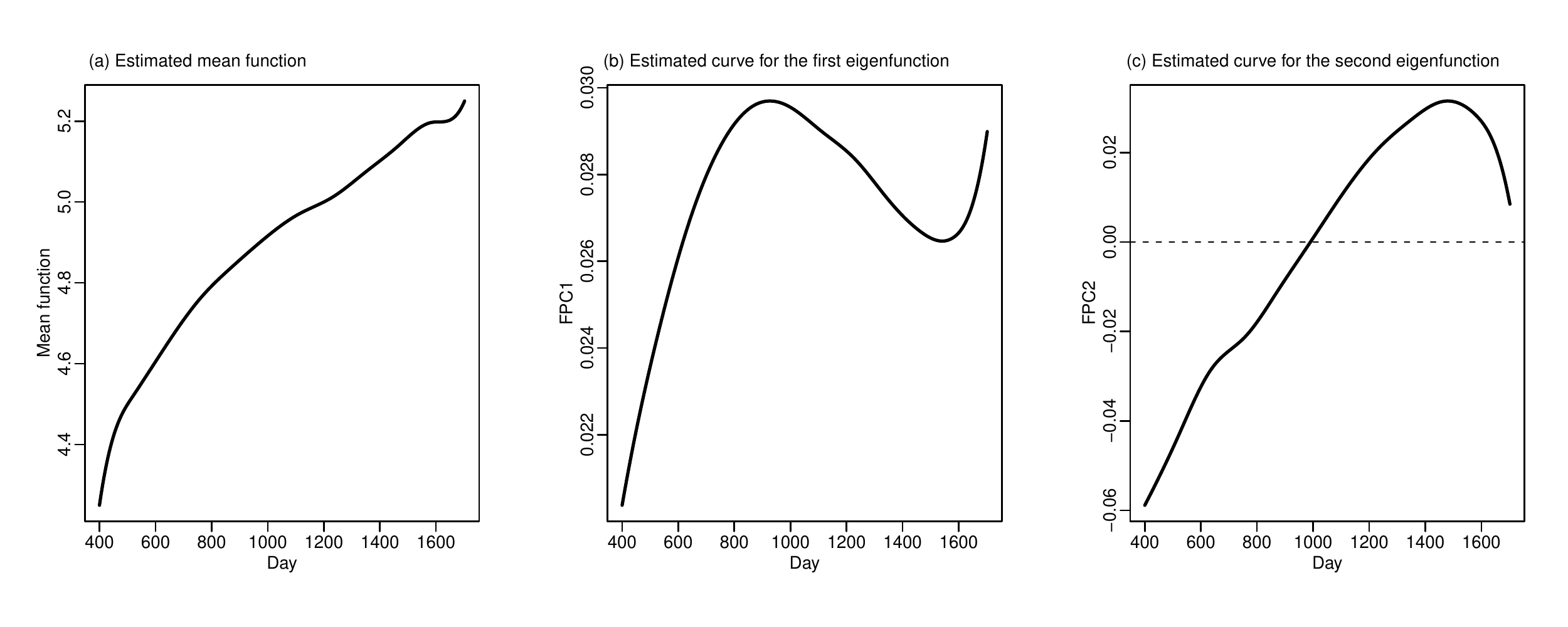}\\
\caption{(a) The estimated mean function for the log-transformed CD4 count data using local linear smoother. (b) The estimation of the first eigenfunction for the log-transformed CD4 count data using the proposed method. (c) The estimation of the second eigenfunction for the log-transformed CD4 count data using the proposed method.}
\label{meaneig}
\end{figure}

To compare our method with \citet{yao2005functional} in a quantitative way, we consider the prediction of the log-transformed CD4 count trajectories. For Kendall functional principal component analysis, the estimated component scores are obtained using the weighted least squares method \citep{chiou2014multivariate}. Then through truncating the infinite series in the Karhunen-Lo\`{e}ve expansion at $L$, the trajectory predictions can be achieved by plugging in the estimated mean function, eigenfunctions and component scores. For comparison, we randomly divide the data set into training set and test set, with $80\%$ data in the former and the remaining $20\%$ in the latter. Denote $\wh{X}_{i}(t_{ij})$ as the prediction for the $i$th subject at the $j$th time point by taking $L = 2$. We define the mean square error between $\wh{X}_{i}(t_{ij})$ and the observed $Y_{ij}$ as
\begin{align}
\frac{1}{N} \sum_{i = 1}^N \frac{1}{m_i} \sum_{j = 1}^{m_i} \big \{ Y_{ij} - \wh{X}_i(t_{ij}) \big \}^2, \nonumber
\end{align}
where $m_i$ is the observation number for the $i$th subject and $N$ is the size of the considered sample. Here for the identifiability of the weighted least squares method, we only consider the sample with more than two observations when computing the mean square error. Table \ref{realMSE} shows the average mean square error of these two methods in both training set and test set for $100$ runs. We reset the training set and test set in each run. It is obvious that our method gains much smaller average mean square error than \citet{yao2005functional} in both training set and test set, which proves the superiority of Kendall functional principal component analysis in dealing with non-Gaussian data in practical.

\begin{table}[htbp]
\caption{Average mean square error in both training set and test set for 100 runs, with the standard errors in parentheses}
  \label{realMSE}
  \begin{center}
  \begin{tabular}{ccc}
  \hline
   &PACE&KFPCA\\
  \hline
  Training set&0.2370(0.0515)&0.1187(0.0132)\\
  Test set&0.3361(0.1104)&0.1273(0.0404)\\
  \hline
  \end{tabular}
  \end{center}
\end{table}

\section{ Conclusion and Discussion }\label{discussion}

Here are some possible extensions to the current work. First, we notice that multivariate functional data and multidimensional functional data are of great interest recent years \citep{happ2018multivariate, wang2020low, virta2020independent}. We only consider univariate functional data in this paper and haven't involved those more complex situations yet. We suppose the generalization of our method to the multivariate/multidimensional functional data is direct and easy to be realized. Second, as pointed out in Remark \ref{remarkrawkend}, we may neglect some comparisons in computing the ``raw" Kendall's $\tau$ covariances. These neglects happen frequently when facing extremely sparse data or imbalanced data for which the observations of distinct subjects are either too sparse or intensely dense. Therefore, some enhancements can be imposed on the estimation procedure for those intractable data. Third, only the eigenfunctions of the Kendall's $\tau$ function are taken into account and more exploration of its eigenvalues may be helpful in some directions. To give an example, the selection of the number of principal components is a long-standing problem in functional principal component analysis. The most convenient way is to use the eigenvalues of the covariance function. Here investigating whether the eigenvalues of Kendall's $\tau$ function can be employed and give a more robust choice is of significance.

\appendix
\section{ Supplementary Material }

\subsection{Proof of Theorem \ref{theorykendall}}

\begin{proof}[Proof of Theorem \ref{theorykendall}]

Let $K$ be the operator with the kernel $K(s, t)$ and $\otimes$ be the tensor product on the functional space $H$, then
\begin{align}
K = E \Big \{ \frac{(X - \wt{X}) \otimes (X - \wt{X})}{\|X - \wt{X}\|^2} \Big \}. \nonumber
\end{align}
As
\begin{align}
X(t) &= \mu(t) + \sum_{k = 1} Z_k \lambda_k^{1/2} \phi_k(t), \nonumber \\
\wt{X}(t) &= \mu(t) + \sum_{k = 1} \wt{Z}_k \lambda_k^{1/2} \phi_k(t), \nonumber
\end{align}
we have $X(t) - \wt{X}(t) = \sum_{k =1} (Z_k - \wt{Z}_k) \lambda_k^{1/2} \phi_k(t)$, where $Z_k$ and $\wt{Z}_k, k = 1, 2, \ldots$ are independent with zero mean and unit variance. Further, we can get
\begin{align}
&\big \{ (X - \wt{X}) \otimes (X - \wt{X}) \big \} \phi_k \nonumber \\
=& \big \{ \sum_j (Z_j - \wt{Z}_j)\lambda_j^{1/2} \phi_j \big \} \otimes \big \{ \sum_l (Z_l - \wt{Z}_l)\lambda_l^{1/2} \phi_l \big \} \phi_k \nonumber \\
=& \sum_j \sum_l (Z_j - \wt{Z}_j)(Z_l - \wt{Z}_l) \lambda_j^{1/2} \lambda_l^{1/2} (\phi_j \otimes \phi_l) \phi_k \nonumber \\
=& \sum_j \sum_l (Z_j - \wt{Z}_j)(Z_l - \wt{Z}_l) \lambda_j^{1/2} \lambda_l^{1/2} \int \phi_j(s) \phi_l(t) \phi_k(t) dt \nonumber \\
=&\lambda_k^{1/2} (Z_k - \wt{Z}_k) \Big \{ \sum_j (Z_j - \wt{Z}_j) \lambda_j^{1/2} \phi_j \Big \}, \label{thm1}
\end{align}
the last equality is due to the orthogonality of the eigenfunctions.
On the other hand, since $Z_k$ and $\wt{Z}_k$ are identically distributed, $Z_k - \wt{Z}_k$ and $\wt{Z}_k - Z_k$ have the same distribution. Then for $k \neq j$,
\begin{align}
E \Big \{ \frac{(Z_k - \wt{Z}_k)(Z_j - \wt{Z}_j)}{\sum_{i} \lambda_i (Z_i - \wt{Z}_i)^2} \Big \} = E \Big \{ \frac{(\wt{Z}_k - Z_k)(Z_j - \wt{Z}_j)}{\sum_{i} \lambda_i (Z_i - \wt{Z}_i)^2} \Big \} = -E \Big \{ \frac{(Z_k - \wt{Z}_k)(Z_j - \wt{Z}_j)}{\sum_{i} \lambda_i (Z_i - \wt{Z}_i)^2} \Big \}. \nonumber
\end{align}
Thus
\begin{align}
E \Big \{ \frac{(Z_k - \wt{Z}_k)(Z_j - \wt{Z}_j)}{\|X - \wt{X}\|^2} \Big \} = E \Big \{ \frac{(Z_k - \wt{Z}_k)(Z_j - \wt{Z}_j)}{\sum_{i} \lambda_i (Z_i - \wt{Z}_i)^2} \Big \} = 0. \label{thm2}
\end{align}
Hence, by the results of (\ref{thm1}) and (\ref{thm2}),
\begin{align}
K \phi_k &= E \Big \{ \frac{(X - \wt{X}) \otimes (X - \wt{X})}{\|X - \wt{X}\|^2} \Big \} \phi_k \nonumber \\
&= \lambda_k^{1/2} E \Big [ \frac{(Z_k - \wt{Z}_k) \big \{ \sum_{j =1} (Z_j - \wt{Z}_j) \lambda_j^{1/2} \phi_j \big \}}{\|X - \wt{X}\|^2} \Big ] \nonumber \\
&= \lambda_k E \Big \{ \frac{(Z_k - \wt{Z}_k)^2}{\|X - \wt{X}\|^2} \Big \} \phi_k. \nonumber
\end{align}
Then $\phi_k$ is the eigenfunction of $K$ corresponding to $\lambda_k^{*}$, where $\lambda_k^{*} = \lambda_k E \{ (Z_k - \wt{Z}_k)^2/\|X - \wt{X}\|^2 \}$, so $\lambda_k^{*}$ has the same sign as $\lambda_k$. Further, for any $i > j$,
\begin{align}
\frac{\lambda_i^{*}}{\lambda_j^{*}} &= E \Big \{ \frac{\lambda_i(Z_i - \wt{Z}_i)^2}{\lambda_i(Z_i - \wt{Z}_i)^2 + \lambda_j(Z_j - \wt{Z}_j)^2 + C} \Big \} \Bigg / E \Big \{ \frac{\lambda_j(Z_j - \wt{Z}_j)^2}{\lambda_i(Z_i - \wt{Z}_i)^2 + \lambda_j(Z_j - \wt{Z}_j)^2 + C} \Big \} \nonumber \\
&< E \Big \{ \frac{\lambda_i(Z_i - \wt{Z}_i)^2}{\lambda_i(Z_i - \wt{Z}_i)^2 + \lambda_i(Z_j - \wt{Z}_j)^2 + C} \Big \} \Bigg / E \Big \{ \frac{\lambda_j(Z_j - \wt{Z}_j)^2}{\lambda_j(Z_i - \wt{Z}_i)^2 + \lambda_j(Z_j - \wt{Z}_j)^2 + C} \Big \} \nonumber \\
&= E \Big \{ \frac{(Z_i - \wt{Z}_i)^2}{(Z_i - \wt{Z}_i)^2 + (Z_j - \wt{Z}_j)^2 + C/\lambda_i} \Big \} \Bigg / E \Big \{ \frac{(Z_j - \wt{Z}_j)^2}{(Z_i - \wt{Z}_i)^2 + (Z_j - \wt{Z}_j)^2 + C/\lambda_j} \Big \} \nonumber \\
&= E \Big \{ \frac{(Z_i - \wt{Z}_i)^2}{(Z_i - \wt{Z}_i)^2 + (Z_j - \wt{Z}_j)^2 + C/\lambda_i} \Big \} \Bigg / E \Big \{ \frac{(Z_i - \wt{Z}_i)^2}{(Z_i - \wt{Z}_i)^2 + (Z_j - \wt{Z}_j)^2 + C/\lambda_j} \Big \} < 1, \nonumber
\end{align}
where $C = \sum_{k \neq i, j}^p \lambda_k(Z_k - \wt{Z}_k)^2$ and the last equality is achieved by the exchangeability of $Z_i - \wt{Z}_i$ and $Z_j - \wt{Z}_j$ as they have the same distribution. The proof is completed.

\end{proof}

\subsection{ Proofs of Theorem \ref{consistency} and Theorem \ref{eigenfunction} }

\begin{lemma}\label{integ}
Let $f$ be a function on interval $\mathcal{T}$ and $t_1, \ldots, t_{\mathcal{M}}$ are equally spaced on $\mathcal{T}$. Assume that the absolute value of the first derivative of $f$ is bounded. Then
\begin{align}
\Bigg | \frac{|\mathcal{T}|}{\mathcal{M}} \sum_{q = 1}^{\mathcal{M}} f(t_q) - \int_{\mathcal{T}} f(t)dt \Bigg | = O \big (\mathcal{M}^{-1} \big ). \nonumber
\end{align}
\end{lemma}

\begin{proof}[Proof of Lemma \ref{integ}]

Denote the right Riemann sum of $f$ as $A_{right}$ and $\{ \mathcal{T}_k \}_{k = 1}^{\mathcal{M} - 1}$ as the partition of $\mathcal{T}$. Then $|\mathcal{T}_k| = |\mathcal{T}|/(\mathcal{M} - 1)$ and $A_{right} = |\mathcal{T}| \sum_{q = 2}^{\mathcal{M}} f(t_q)/(\mathcal{M} - 1)$. By Taylor's expansion, we have
\begin{align}
\int_{\mathcal{T}_k} f(t) dt &= \int_{\mathcal{T}_k} \Big \{ f(t_{k + 1}) + O(|\mathcal{T}_k|) \Big \} dt \nonumber \\
&= f(t_{k + 1}) |\mathcal{T}_k| + O(|\mathcal{T}_k|^2) \nonumber \\
&= \frac{|\mathcal{T}|}{\mathcal{M} - 1} f(t_{k + 1}) + O \big (\mathcal{M}^{-2} \big ). \nonumber
\end{align}
Therefore,
\begin{align}
\int_{\mathcal{T}} f(t) dt = \sum_{k = 1}^{\mathcal{M} - 1} \int_{\mathcal{T}_k} f(t) dt = \frac{|\mathcal{T}|}{\mathcal{M} - 1} \sum_{q = 2}^{\mathcal{M}} f(t_q) + O \big (\mathcal{M}^{-1} \big ) = A_{right} + O \big (\mathcal{M}^{-1} \big ). \nonumber
\end{align}
Further,
\begin{align}
\Big | A_{right} - \int_{\mathcal{T}} f(t)dt \Big | = O \big (\mathcal{M}^{-1} \big ). \nonumber
\end{align}
As
\begin{align}
\frac{|\mathcal{T}|}{\mathcal{M}} \sum_{q = 1}^{\mathcal{M}} f(t_q) = \frac{\mathcal{M} - 1}{\mathcal{M}} A_{right} + \frac{|\mathcal{T}|}{\mathcal{M}} f(t_1), \nonumber
\end{align}
then
\begin{align}
\Bigg | \frac{|\mathcal{T}|}{\mathcal{M}} \sum_{q = 1}^{\mathcal{M}} f(t_q) - \int_{\mathcal{T}} f(t)dt \Bigg | &= \Bigg | \frac{\mathcal{M} - 1}{\mathcal{M}} A_{right} + \frac{|\mathcal{T}|}{\mathcal{M}} f(t_1) - \int_{\mathcal{T}} f(t)dt \Bigg | \nonumber \\
&\leq \frac{\mathcal{M} - 1}{\mathcal{M}} \Big | A_{right} - \int_{\mathcal{T}} f(t)dt \Big | + \frac{1}{\mathcal{M}} \Big | |\mathcal{T}| f(t_1) - \int_{\mathcal{T}} f(t)dt \Big | \nonumber \\
&= O \big (\mathcal{M}^{-1} \big ). \nonumber
\end{align}
The proof is completed.

\end{proof}

\begin{proof}[Proof of Theorem \ref{consistency}]

Define
\begin{align}
K_i^I(t_{ik}, t_{il}) &= \frac{1}{N-1} |\mathcal{T}| \sum_{j \neq i} \Big [ \frac{\big \{ X_{i}(t_{ik}) - X_{j}(t_{ik}) \big \} \big \{ X_{i}(t_{il}) - X_{j}(t_{il}) \big \}}{\int_{\mathcal{T}} \big \{ X_i(u) - X_j(u) \big \}^2 du} \Big ]. \nonumber
\end{align}

\noindent \textbf{(1)} The relationship between $K_i(t_{ik}, t_{il})$ and $K_i^I(t_{ik}, t_{il})$

We first explore the relationship between $K_i(t_{ik}, t_{il})$ and $K_i^I(t_{ik}, t_{il})$. For this purpose, we further define
\begin{align}
{K}_i^X(t_{ik}, t_{il}) &= \frac{1}{N-1} \sum_{j \neq i} \Big [ \frac{\big \{ X_{i}(t_{ik}) - X_{j}(t_{ik}) \big \} \big \{ X_{i}(t_{il}) - X_{j}(t_{il}) \big \}}{\sum_{q = 1}^{m_i} \big \{ X_i(t_{iq}) - X_j(t_{iq}) \big \}^2 / m_i} \Big ], \nonumber \\
{\wt{K}}_i^X(t_{ik}, t_{il}) &= \frac{1}{N-1} \sum_{j \neq i} \Big [ \frac{\big \{ X_{i}(t_{ik}) - X_{j}(t_{ik}) \big \} \big \{ X_{i}(t_{il}) - X_{j}(t_{il}) \big \}}{\sum_{q = 1}^{m_i} \big \{ X_i(t_{q}) - X_j(t_{q}) \big \}^2 / m_i} \Big ]. \nonumber
\end{align}

\noindent \textbf{(i)} $K_i(t_{ik}, t_{il})$ and ${K}_i^X(t_{ik},t_{il})$

Define functions $f_i(\textbf{x}) = \frac{x_1 x_2}{\sum_{k = 1}^{m_i} x_k^2 / m_i}, i = 1, \ldots, N$, where $\textbf{x} = (x_1, \ldots, x_{m_i})^\top \in \mathbb{R}^{m_i}$. Then
\begin{align}
\frac{\big \{ Y_{ik} - \wh{X}_{j}(t_{ik}) \big \} \big \{ Y_{il} - \wh{X}_{j}(t_{il}) \big \}}{\sum_{q = 1}^{m_i} \{ Y_{iq} - \wh{X}_j(t_{iq})\}^2 / m_i} &= f_i(Y_{ik} - \wh{X}_{j}(t_{ik}), Y_{il} - \wh{X}_{j}(t_{il}), \wh{A}_{ij}), \nonumber \\
\frac{\big \{ X_{i}(t_{ik}) - X_{j}(t_{ik}) \big \} \big \{ X_{i}(t_{il}) - X_{j}(t_{il}) \big \}}{\sum_{q = 1}^{m_i} \{X_i(t_{iq}) - X_j(t_{iq})\}^2 / m_i} &= f_i(X_{i}(t_{ik}) - X_{j}(t_{ik}), X_{i}(t_{il}) - X_{j}(t_{il}), A_{ij}), \nonumber
\end{align}
where $\wh{A}_{ij} = \big (Y_{i1} - \wh{X}_{j}(t_{i1}), \ldots, Y_{im_i} - \wh{X}_{j}(t_{im_i}) \big )^\top \in \mathbb{R}^{m_i-2}$ and $A_{ij} = \big (X_{i}(t_{i1}) - X_{j}(t_{i1}), \ldots, X_{i}(t_{im_i}) - X_{j}(t_{im_i}) \big )^\top \in \mathbb{R}^{m_i-2}$. We Taylor-expand $f_i(Y_{ik} - \wh{X}_{j}(t_{ik}), Y_{il} - \wh{X}_{j}(t_{il}), \wh{A}_{ij})$ about $f_i(X_{i}(t_{ik}) - X_{j}(t_{ik}), X_{i}(t_{il}) - X_{j}(t_{il}), A_{ij})$, obtaining
\begin{align}
&f_i(Y_{ik} - \wh{X}_{j}(t_{ik}), Y_{il} - \wh{X}_{j}(t_{il}), \wh{A}_{ij}) \nonumber \\
= &f_i(X_{i}(t_{ik}) - X_{j}(t_{ik}), X_{i}(t_{il}) - X_{j}(t_{il}), A_{ij}) + \nabla f_i(\textbf{x}_{ij})^\top \left ( \begin{array}{c}
\epsilon_{ik} - \wh{X}_{j}(t_{ik}) + X_j(t_{ik}) \\
\epsilon_{il} - \wh{X}_{j}(t_{il}) + X_j(t_{il}) \\
\ \vdots \\
\epsilon_{im_i} - \wh{X}_{j}(t_{im_i}) + X_j(t_{im_i})
\end{array} \right ), \nonumber
\end{align}
where $\textbf{x}_{ij} \in \mathbb{R}^{m_i}$ lies between $\big (Y_{ik} - \wh{X}_{j}(t_{ik}), Y_{il} - \wh{X}_{j}(t_{il}), \wh{A}_{ij} \big )^\top$ and $\big (X_{i}(t_{ik}) - X_{j}(t_{ik}), X_{i}(t_{il}) - X_{j}(t_{il}), A_{ij} \big )^\top$. Thereby
\begin{align}
K_i(t_{ik},t_{il}) = {K}_i^X(t_{ik}, t_{il}) + \frac{1}{N-1} \sum_{j \neq i} \nabla f_i(\textbf{x}_{ij})^\top \left ( \begin{array}{c}
\epsilon_{ik} - \wh{X}_{j}(t_{ik}) + X_j(t_{ik}) \\
\epsilon_{il} - \wh{X}_{j}(t_{il}) + X_j(t_{il}) \\
\ \vdots \\
\epsilon_{im_i} - \wh{X}_{j}(t_{im_i}) + X_j(t_{im_i})
\end{array} \right ). \label{i-1}
\end{align}

We then think about $ \epsilon_{ik} - \wh{X}_{j}(t_{ik}) + X_j(t_{ik}) $. We have
\begin{align}
&\epsilon_{ik} - \wh{X}_{j}(t_{ik}) + X_j(t_{ik}) \nonumber \\
=& \epsilon_{ik} - \sum_q^{m_j} \frac{ k_{h^{'}}(t_{ik} - t_{jq})}{\sum_p^{m_j} k_{h^{'}}(t_{ik} - t_{jp})} \big \{ Y_{jq} - X_j(t_{ik}) \big \} \nonumber \\
=& \epsilon_{ik} - \sum_q^{m_j} \frac{ k_{h^{'}}(t_{ik} - t_{jq})}{\sum_p^{m_j} k_{h^{'}}(t_{ik} - t_{jp})} \big \{ X_j(t_{jq}) - X_j(t_{ik}) + \epsilon_{jq} \big \} \nonumber \\
=& \sum_q^{m_j} \frac{ k_{h^{'}}(t_{ik} - t_{jq})}{\sum_p^{m_j} k_{h^{'}}(t_{ik} - t_{jp})} \big \{ X_j(t_{ik}) - X_j(t_{jq}) + \epsilon_{ik} - \epsilon_{jq} \big \} \nonumber \\
=& \sum_q^{m_j} w_{ijkq} \big \{ X_j(t_{ik}) - X_j(t_{jq}) \big \} + \sum_q^{m_j} w_{ijkq} \big ( \epsilon_{ik} - \epsilon_{jq} \big ), \label{i-odd}
\end{align}
where $w_{ijkq} = k_{h^{'}}(t_{ik} - t_{jq})/\sum_p^{m_j} k_{h^{'}}(t_{ik} - t_{jp})$. Further,
\begin{align}
\Big | \sum_q^{m_j} w_{ijkq} \big \{ X_j(t_{ik}) - X_j(t_{jq}) \big \} \Big | &\leq \sum_q^{m_j} w_{ijkq} \big | X_j(t_{ik}) - X_j(t_{jq}) \big | \nonumber \\
& \leq C \sum_q^{m_j} w_{ijkq} \big | t_{ik} - t_{jq} \big | \leq Ch^{'}, \label{i-odd1}
\end{align}
the second and third inequalities are obtained from Assumption \ref{Xder} and Assumption \ref{kernel} respectively. As $\nabla f_i$ is bounded, combining (\ref{i-1}), (\ref{i-odd}) and (\ref{i-odd1}), we have
\begin{align}
K_i(t_{ik},t_{il}) &= {K}_i^X(t_{ik}, t_{il}) + O(h^{'}) + \frac{1}{N-1} \sum_{j \neq i} \nabla f_i(\textbf{x}_{ij})^\top \left ( \begin{array}{c}
\sum_q^{m_j} w_{ijkq} \big ( \epsilon_{ik} - \epsilon_{jq} \big ) \\
\sum_q^{m_j} w_{ijlq} \big ( \epsilon_{il} - \epsilon_{jq} \big ) \\
\ \vdots \\
\sum_q^{m_j} w_{ijm_iq} \big ( \epsilon_{im_i} - \epsilon_{jq} \big )
\end{array} \right ) \nonumber \\
&\triangleq {K}_i^X(t_{ik}, t_{il}) + O(h^{'}) + E_{ikl}, \label{i}
\end{align}
where $E(E_{ikl}) = 0$.

\noindent \textbf{(ii)} ${K}_i^X(t_{ik}, t_{il})$ and ${\wt{K}}_i^X(t_{ik}, t_{il})$

Define functions $h_i(x_1, x_2, \textbf{y}) = \frac{x_1 x_2}{\sum_{k = 1}^{m_i} y_k^2 / m_i}, i = 1, \ldots, N$, where $\textbf{y} = (y_1, \ldots, y_{m_i})^\top \in \mathbb{R}^{m_i}$. Then by Taylor expansion, we can get
\begin{align}
{K}^X_i(t_{ik},t_{il}) = {\wt{K}}^X_i(t_{ik}, t_{il}) + \frac{1}{N-1} \sum_{j \neq i} \nabla h_i(\textbf{z}_{ij})^\top \left ( \begin{array}{c}
0 \\
0 \\
\Delta_{ij1}\\
\ \vdots \\
\Delta_{ijm_i}
\end{array} \right ), \nonumber
\end{align}
where $\textbf{z}_{ij}$ lies between the corresponding values and
\begin{align}
\Delta_{ijq} &= \big \{ X_i(t_{iq}) - X_j(t_{iq}) \big \} - \big \{ X_i(t_{q}) - X_j(t_{q}) \big \} \nonumber \\
&= \big \{ X_i(t_{iq}) - X_i(t_{q}) \big \} - \big \{ X_j(t_{iq}) - X_j(t_{q}) \big \} \nonumber \\
&= O_p(t_{iq} - t_{q}) = O_p(\Delta t) = o_p(1), \nonumber
\end{align}
the third equality is derived from Assumption \ref{Xder}, while the last two equalities are obtained from Assumption \ref{deltat}.
Since $\nabla h_i$ is bounded, we have
\begin{align}\label{ii}
{K}^X_i(t_{ik},t_{il}) = {\wt{K}}^X_i(t_{ik}, t_{il}) + o_p(1).
\end{align}

\noindent \textbf{(iii)} ${\wt{K}}_i^X(t_{ik}, t_{il})$ and $K_i^I(t_{ik}, t_{il})$

According to Lemma \ref{integ},
\begin{align}
{\wt{K}}_i^X(t_{ik}, t_{il}) &= \frac{1}{N-1} |\mathcal{T}| \sum_{j \neq i} \Big [ \frac{\big \{ X_{i}(t_{ik}) - X_{j}(t_{ik}) \big \} \big \{ X_{i}(t_{il}) - X_{j}(t_{il}) \big \}}{ (|\mathcal{T}| / m_i) \sum_{q = 1}^{m_i} \big \{ X_i(t_{q}) - X_j(t_{q}) \big \}^2} \Big ] \nonumber \\
&= \frac{1}{N-1} |\mathcal{T}| \sum_{j \neq i} \Big [ \frac{\big \{ X_{i}(t_{ik}) - X_{j}(t_{ik}) \big \} \big \{ X_{i}(t_{il}) - X_{j}(t_{il}) \big \}}{ \int \big \{ X_{i}(u) - X_{j}(u) \big \}^2 du + O_p(m_i^{-1}) } \Big ] \nonumber \\
&= \frac{1}{N-1} |\mathcal{T}| \sum_{j \neq i} \Big [ \frac{\big \{ X_{i}(t_{ik}) - X_{j}(t_{ik}) \big \} \big \{ X_{i}(t_{il}) - X_{j}(t_{il}) \big \}}{ \int \big \{ X_{i}(u) - X_{j}(u) \big \}^2 du } + O_p(m_i^{-1}) \Big ] \nonumber \\
&= \frac{1}{N-1} |\mathcal{T}| \sum_{j \neq i} \Big [ \frac{\big \{ X_{i}(t_{ik}) - X_{j}(t_{ik}) \big \} \big \{ X_{i}(t_{il}) - X_{j}(t_{il}) \big \}}{ \int \big \{ X_{i}(u) - X_{j}(u) \big \}^2 du } \Big ] + O_p(m_i^{-1}) \nonumber \\
&= \frac{1}{N-1} |\mathcal{T}| \sum_{j \neq i} \Big [ \frac{\big \{ X_{i}(t_{ik}) - X_{j}(t_{ik}) \big \} \big \{ X_{i}(t_{il}) - X_{j}(t_{il}) \big \}}{ \int \big \{ X_{i}(u) - X_{j}(u) \big \}^2 du } \Big ] + O_p(M^{-1}) \nonumber \\
&= K_i^I(t_{ik}, t_{il}) + O_p(M^{-1}). \label{iii}
\end{align}

Combining (\ref{i}), (\ref{ii}) and (\ref{iii}), we get
\begin{align}\label{part1}
K_i(t_{ik}, t_{il}) = K_i^I(t_{ik}, t_{il}) + E_{ikl} + O_p \Big (M^{-1} + h^{'} \Big )
\end{align}

\noindent \textbf{(2)} Consistency

Based on straightforward calculation, we can get
\begin{align}
\wh{K}_0(s, t) = (\mathcal{A}_1K_{00} - \mathcal{A}_2K_{10} -\mathcal{A}_3K_{01}) \mathcal{B}^{-1}, \nonumber
\end{align}
where
\begin{align}
\mathcal{A}_1 &= S_{20}S_{02} - S_{11}^2, \mathcal{A}_2 = S_{10}S_{02} - S_{01}S_{11}, \mathcal{A}_3 = S_{01}S_{20} - S_{10}S_{11}, \nonumber \\
\mathcal{B} &= \mathcal{A}_1S_{00} - \mathcal{A}_2S_{10} - \mathcal{A}_3S_{01}, \nonumber \\
S_{pq} &= \frac{1}{n}\sum_{i = 1}^N \sum_{1 \leq k \neq l \leq m_i} \Big ( \frac{t_{ik} - s}{h} \Big )^p \Big ( \frac{t_{il} - t}{h} \Big )^q \kappa_h \big ( t_{ik} - s \big ) \kappa_h \big ( t_{il} - t \big ), \nonumber \\
K_{pq} &= \frac{1}{n}\sum_{i = 1}^N \sum_{1 \leq k \neq l \leq m_i} K_i(t_{ik}, t_{il}) \Big ( \frac{t_{ik} - s}{h} \Big )^p \Big ( \frac{t_{il} - t}{h} \Big )^q \kappa_h \big ( t_{ik} - s \big ) \kappa_h \big ( t_{il} - t \big ), \nonumber \\
n &= \sum_{i = 1}^N m_i(m_i - 1). \nonumber
\end{align}
Then
\begin{align}
(\wh{K}_0 - K_0)(s, t) = (\mathcal{A}_1K_{00}^{*} - \mathcal{A}_2K_{10}^{*} - \mathcal{A}_3K_{01}^{*}) \mathcal{B}^{-1}, \label{k0diff}
\end{align}
where
\begin{align}
K_{pq}^{*} = K_{pq} - K_0(s, t)S_{pq} - hK_0^{(1, 0)}(s, t)S_{p + 1, q} - hK_0^{(0, 1)}(s, t)S_{p, q+ 1}. \nonumber
\end{align}
Specifically,
\begin{align}
K_{00}^{*} = \frac{1}{n}\sum_{i = 1}^N \sum_{1 \leq k \neq l \leq m_i} &\big \{ K_i(t_{ik}, t_{il}) - K_0(s, t) - K_0^{(1, 0)}(s, t)(t_{ik} - s) - K_0^{(0, 1)}(s, t)(t_{il} - t) \big \} \nonumber \\
&\times \kappa_h \big ( t_{ik} - s \big ) \kappa_h \big ( t_{il} - t \big ). \nonumber
\end{align}
By (\ref{part1}) and Taylor's expansion,
\begin{align}
&K_i(t_{ik}, t_{il}) - K_0(s, t) - K_0^{(1, 0)}(s, t)(t_{ik} - s) - K_0^{(0, 1)}(s, t)(t_{il} - t) \nonumber \\
=& K_i^I(t_{ik}, t_{il}) + E_{ikl} + O_p \Big (M^{-1} + h^{'} \Big ) - K_0(s, t) - K_0^{(1, 0)}(s, t)(t_{ik} - s) - K_0^{(0, 1)}(s, t)(t_{il} - t) \nonumber \\
=& K_i^I(t_{ik}, t_{il}) + E_{ikl} - K_0(t_{ik}, t_{il}) + \big \{K_0(t_{ik}, t_{il})- K_0(s, t) \nonumber \\
& \qquad - K_0^{(1, 0)}(s, t)(t_{ik} - s) - K_0^{(0, 1)}(s, t)(t_{il} - t) \big \} + O_p \Big (M^{-1} + h^{'} \Big ) \nonumber \\
=& \epsilon_{ikl}^{*} + O_p(h^2 + M^{-1} + h^{'}), \nonumber
\end{align}
where $\epsilon_{ikl}^{*} = K_i^I(t_{ik}, t_{il}) - K_0(t_{ik}, t_{il}) + E_{ikl}$ and it is easy to observe that $E(\epsilon_{ikl}^{*}) = 0$.
Then we have
\begin{align}
K_{00}^{*} = \frac{1}{n}\sum_{i = 1}^N \sum_{1 \leq k \neq l \leq m_i} \epsilon_{ikl}^{*} \kappa_h \big ( t_{ik} - s \big ) \kappa_h \big ( t_{il} - t \big ) + O_p(h^2 + M^{-1} + h^{'}). \nonumber
\end{align}
Since $E(\epsilon_{ikl}^{*}) = 0$, by classical uniform convergence rates of kernel smoother \citep{li2011efficient, hansen2008uniform, masry1996multivariate},
\begin{align}
K_{00}^{*} = O_p \big (\{\log n / (nh^2)\}^{1/2} + h^2 + M^{-1} + h^{'} \big ). \label{K00}
\end{align}
In the same way, we can get
\begin{align}
K_{10}^{*} = O_p \big ( \{ \log n / (nh^2)\}^{1/2} + h^2 + M^{-1} + h^{'} \big ), \label{K10}\\
K_{01}^{*} = O_p \big ( \{ \log n / (nh^2)\}^{1/2} + h^2 + M^{-1} + h^{'} \big ). \label{K01}
\end{align}
Furthermore, we have
\begin{align}
S_{00} &= g(s)g(t) + O_p \big (\{\log n / (nh^2)\}^{1/2} + h^2 \big ), \nonumber \\
S_{10} &= O_p\big (\{\log n / (nh^2)\}^{1/2} + h \big ), \quad S_{01} = O_p\big (\{\log n / (nh^2)\}^{1/2} + h \big ), \nonumber \\
S_{20} &= g(s)g(t)\nu_2 + O_p\big (\{\log n / (nh^2)\}^{1/2} + h^2 \big ), \nonumber \\
S_{02} &= g(s)g(t)\nu_2 + O_p\big (\{\log n / (nh^2)\}^{1/2} + h^2 \big ), \nonumber \\
S_{11} &= O_p \big (\{\log n / (nh^2)\}^{1/2} + h^2 \big ). \nonumber
\end{align}
Thus
\begin{align}
\mathcal{A}_1 &= g(s)^2g(t)^2 \nu_2^2 + O_p\big (\{\log n / (nh^2)\}^{1/2} + h^2 \big ), \nonumber \\
\mathcal{A}_2 &= O_p\big (\{\log n / (nh^2)\}^{1/2} + h \big ), \quad \mathcal{A}_3 = O_p\big (\{\log n / (nh^2)\}^{1/2} + h \big ), \nonumber \\
\mathcal{B} &= g(s)^3g(t)^3 \nu_2^2 + O_p\big (\{\log n / (nh^2)\}^{1/2} + h \big ). \nonumber
\end{align}
Then
\begin{align}
\mathcal{A}_1 \mathcal{B}^{-1} &= g(s)^{-1}g(t)^{-1} + O_p\big (\{\log n / (nh^2)\}^{1/2} + h \big ) \label{A1B} \\
\mathcal{A}_2 \mathcal{B}^{-1} &= O_p\big (\{\log n / (nh^2)\}^{1/2} + h \big ) \label{A2B}\\
\mathcal{A}_3 \mathcal{B}^{-1} &= O_p\big (\{\log n / (nh^2)\}^{1/2} + h \big ) \label{A3B}
\end{align}
Combining (\ref{k0diff}) - (\ref{A3B}), we have
\begin{align}
\Big |\wh{K}_0(s, t) - K_0(s, t) \Big | = O_p \big (\{\log n / (nh^2)\}^{1/2} + h^2 + M^{-1} + h^{'} \big ), \forall s, t \in \mathcal{T}. \nonumber
\end{align}
According to Assumption \ref{TYm}, $n = \sum_{i = 1}^N m_i(m_i - 1) = O(NM^2)$, then
\begin{align}
\Big |\wh{K}_0(s, t) - K_0(s, t) \Big | = O_p \big (\tau_1 + \tau_2 \big ), \forall s, t \in \mathcal{T}, \nonumber
\end{align}
where $\tau_1 = \big (NM^2h^2/\log N \big )^{-1/2} + h^2$ and $\tau_2 = M^{-1} + h^{'}$. Then (\ref{consisform}) follows.

\end{proof}

For the proof of Theorem \ref{eigenfunction}, as an analogy of Theorem 2 in \citep{yao2005functional}, we can get the convergence rate of the estimated eigenfunctions directly making use of (\ref{consisform}).
More details can be found in \citet{yao2005functional}.

\bibliographystyle{unsrtnat}
\bibliography{ref}

\end{document}